\documentclass[a4paper,11pt]{article}
\usepackage[caption=false,font=footnotesize]{subfig}
\usepackage[T1]{fontenc}
\usepackage[utf8]{inputenc}
\usepackage[english]{babel}

\usepackage[ruled,vlined]{algorithm2e}
\usepackage{url}
\usepackage{calc}
\usepackage{amsmath}
\usepackage{amssymb}
\usepackage{amstext}
\usepackage{mathrsfs}
\usepackage{amsfonts}
\usepackage{dsfont} 
\usepackage{stmaryrd}

\usepackage{color}
\usepackage{multirow}
\usepackage{array}
\usepackage{graphicx}

\usepackage[english]{varioref}

\usepackage[amsmath,thmmarks]{ntheorem}

\theoremsymbol{\ensuremath{\Box}}
\theoremstyle{nonumberplain}
\newtheorem{proof}{Proof}

\theoremsymbol{\ensuremath{\Box}}
\theoremseparator{:}
\theoremstyle{plain}
\newtheorem{definition}{Definition}

\newtheorem{proposition}{Proposition}

\newtheorem{example}{Example}
\newtheorem{remark}{Remark}

\usepackage{color}

\begin{document}
%
\title{Application of Steganography for Anonymity through
the Internet}



\author{Jacques M. Bahi, Jean-Fran\c cois Couchot, Nicolas
Friot, and  Christophe Guyeux*\\
\\
FEMTO-ST Institute, UMR 6174 CNRS\\
Computer Science Laboratory DISC\\
University of Franche-Comt\'{e}\\
Besan\c con, France\\
\\
\{jacques.bahi,  jean-francois.couchot,
nicolas.friot, christophe.guyeux\}@femto-st.fr\\
* Authors in alphabetic order}


%


\maketitle

\begin{abstract}
In this paper, a novel steganographic scheme based on
chaotic iterations is proposed. This research work takes place into the information hiding
security framework. The applications for anonymity and privacy through the Internet are regarded too.
To guarantee such an anonymity, it should be possible
to set up a secret communication channel into a web page, being both secure and robust. 
To achieve this goal, we propose 
an information hiding scheme being stego-secure, which is the highest level of security in a
well defined and studied category of attacks called ``watermark-only attack''.
This category of attacks is the best context to study steganography-based anonymity through the
Internet. The steganalysis of our steganographic process is also studied in
order to show it security in a real test framework.\newline

\end{abstract}

%

\section{Introduction}\label{sec:intro}

In common opinion or for non specialists, anonymity through the Internet is
only desirable for malicious use. A frequent thought is that individuals who search or use anonymity tools
have something wrong or shameful to hide.
Thus, as privacy and anonymity software as proxy or Tor~\cite{www:tor,Clark:2007:UAW:1280680.1280687} 
are only used by terrorists, pedophiles, weapon 
merchants, and so on, such tools should be forbidden.
However, terrorism or pedophilia existed in the absence of the Internet.
Furthermore, recent actualities recall to us that, in 
numerous places around the world, to have an opinion 
that diverges from the one imposed by political or 
religious leaders is something considered as negative, 
suspicious, or illegal.
For instance, Saudi blogger Hamza Kashgari jailed, may face execution after tweets about Muhammad.
Generally speaking, the so-called Arab Spring, and 
current fighting and uncertainty in Syria, have taught
to us the following facts. First, the Internet is a
media of major importance, which is difficult to arrest
or to silence, bearing witness to the need for 
democracy, transparency, and efforts to combat 
corruption. 
Second, claiming his/her opinions, making journalism
or politics, is dangerous in various states, and can lead
to the death penalty (as for numerous Iranian bloggers: Hossein
Derakhshan~\cite{wiki:Hossein-Derakhshan}, Vahid Asghari~\cite{wiki:Vahid-Asghari}, etc.).

Considering that the freedom of expression is a 
fundamental right that must be protected, that 
journalists must be able to inform the community without
risking their own lives, and that to be a defender of
human rights can be dangerous, various software have
emerged these last decades to preserve anonymity or
privacy through the Internet. The most famous tool of
this kind is probably Tor, the onion router.
Tor client software routes Internet traffic through a worldwide volunteer network of servers, in order to conceal an user's location or usage from anyone conducting network surveillance or traffic analysis. 
Another example of this kind is given by Perseus~\cite{www:perseus}, 
a firefox plugin~\cite{www:perseus-firefox-plugin} that protect personal
data, without infringing any national crypto regulations, and that 
preserve the true needs of national security.
Perseus replaces cryptography by coding theory 
techniques, such that only agencies with a strong enough 
computer power can eavesdrop
traffic in an acceptable amount of time.
Finally, anonymous proxy servers 
 around the world can help
to keep machines behind them anonymous:
the destination server (the server that ultimately satisfies the web request) receives requests from the anonymizing proxy server, and thus does not receive information about the end user's address. 

These three solutions are not without flaws.
For instance, when considering anonymizers, 
the requests are not anonymous to the 
anonymizing proxy server, which simply moves the 
problem on: are these proxy servers worthy of trust?
Perseus can be broken with enough computer power. And
due to its central position and particular conception,
 Tor is targeted by numerous attacks and presents 
 various weakness (bad apple attack, or the fact that 
 Tor cannot protect against monitoring of traffic at the 
 boundaries of the Tor network).

Considering these flaws, and because having a variety of
solutions to provide anonymity is a good rule of thumb,
a steganographic approach is often regarded in that
context~\cite{bg10:ip}. Steganography can be applied in several ways
to preserve anonymity through the Internet, encompassing
the creation of secret channels through background 
images of websites, into Facebook photo galleries, on
audio or video streams, or in non-interpreted characters
in HTML source codes.
The authors' intention is not to describe precisely
these well-known techniques, but to explain how to
evaluate their security.
They applied it on a new algorithm of steganography based on chaotic iterations 
and data embedding in least significant coefficients.
This state-of-the-art in information hiding security 
is organized as follows.

In Section
\ref{sec:basic-reminders}, some basic reminders concerning both mathematical
notions and notations, and the Most and Least Significant Coefficients
are given. Our new steganographic process called $\mathcal{DI}_3$ 
which is suitable to guarantee
anonymity of data for privacy on the Internet is presented in Section
\ref{sec:new-process-di-1}.  In Section
\ref{sec:dh-security}, a reminder about information hiding
security is realized. The attacks classification in a steganographic framework are
given, and the level of security of $\mathcal{DI}_3$  is studied. In the next
section the security of our new scheme is evaluated. Then, in Section-
\ref{section:steganalysis} the steganalysis of the proposed process is realized, and it is compared with other steganographic schemes in the literature. This research work ends by a conclusion
section, where our contribution is summarized and intended future researches are presented.

\section{Basic Reminders}\label{sec:basic-reminders}

\subsection{Mathematical definitions and notations}\label{sec:math-def}

Let $S^{n}$ denotes the $n^{th}$ term of a sequence $S$, and $V_{i}$ the $i^{th}$ component of a vector $V$. For $a,b \in \mathds{N}$, we use the following notation: $\llbracket a;b \rrbracket=\{a,a+1,a+2,\hdots,b\}$.

\begin{definition}
\label{def:strategy-adapter-k}
Let $\mathsf{k} \in \mathds{N}^\ast$. 
The set of all sequences  which elements belong into $\llbracket 1; \mathsf{k}
\rrbracket$, called \emph{strategy adapters} on $\llbracket 1; \mathsf{k}
\rrbracket$, is denoted by $\mathds{S}_\mathsf{k}$.
\end{definition}

\begin{definition}
The \emph{support of a finite sequence} $S$ of $n$ terms is the finite set
$\mathscr{S}(S)=\left\{ S^k, k < n \right\}$ containing all the distinct
values of $S$. Its cardinality is s.t. $\#\mathscr{S}(S) \leqslant n$.
\end{definition}

\begin{definition}\label{def:injective-sequence}
A finite sequence $S \in \mathds{S}_\mathsf{N}$ of $n$ terms is \emph{injective}
if $n = \#\mathscr{S}(S)$.
It is \emph{onto}
if $N = \#\mathscr{S}(S)$. Finally, it is bijective if and only if it is both
injective and onto, so $n=N = \#\mathscr{S}(S)$.
\end{definition}

\begin{remark}
On the one hand, ``$S$ is injective'' reflects the fact that  all the $n$ terms of
the sequence $S$ are distinct. On the other hand, ``$S$ is onto'' means
that all the values of the set $\llbracket 1;\mathsf{N}\rrbracket$
are reached at least once.
\end{remark}
\subsection{The Most and Least Significant
Coefficients}\label{sec:msc-lsc}

We first notice that terms of the original content $x$ that may be replaced by terms issued
from the watermark $y$ are less important than other: they could be changed 
without be perceived as such. More generally, a 
\emph{signification function} 
attaches a weight to each term defining a digital media,
depending on its position $t$.

\begin{definition}
A \emph{signification function} is a real sequence 
$(u^k)^{k \in \mathds{N}}$. 
\end{definition}

\begin{example}\label{Exemple LSC}
Let us consider a set of    
grayscale images stored into portable graymap format (P3-PGM):
each pixel ranges between 256 gray levels, \textit{i.e.},
is memorized with eight bits.
In that context, we consider 
$u^k = 8 - (k  \mod  8)$  to be the $k$-th term of a signification function 
$(u^k)^{k \in \mathds{N}}$. 
Intuitively, in each group of eight bits (\textit{i.e.}, for each pixel) 
the first bit has an importance equal to 8, whereas the last bit has an
importance equal to 1. This is compliant with the idea that
changing the first bit affects more the image than changing the last one.
\end{example}

\begin{definition}
\label{def:msc-lsc}
Let $(u^k)^{k \in \mathds{N}}$ be a signification function, 
$m$ and $M$ be two reals s.t. $m < M$. 
\begin{itemize}
\item The \emph{most significant coefficients (MSCs)} of $x$ is the finite 
  vector  $$u_M = \left( k ~ \big|~ k \in \mathds{N} \textrm{ and } u^k 
    \geqslant M \textrm{ and }  k \le \mid x \mid \right);$$
 \item The \emph{least significant coefficients (LSCs)} of $x$ is the 
finite vector 
$$u_m = \left( k ~ \big|~ k \in \mathds{N} \textrm{ and } u^k 
  \le m \textrm{ and }  k \le \mid x \mid \right);$$
 \item The \emph{passive coefficients} of $x$ is the finite vector 
   $$u_p = \left( k ~ \big|~ k \in \mathds{N} \textrm{ and } 
u^k \in ]m;M[ \textrm{ and }  k \le \mid x \mid \right).$$
 \end{itemize}
 \end{definition}

For a given host content $x$,
MSCs are then ranks of $x$  that describe the relevant part
of the image, whereas LSCs translate its less significant parts.

\begin{example}
These two definitions are illustrated on Figure~\ref{fig:MSCLSC}, where the
significance function $(u^k)$ is defined as in Example \ref{Exemple LSC}, $M=5$,
and $m=6$.

\begin{figure}[htb]

\begin{minipage}[b]{.98\linewidth}
  \centering
  \centerline{\includegraphics[width=2.5cm]{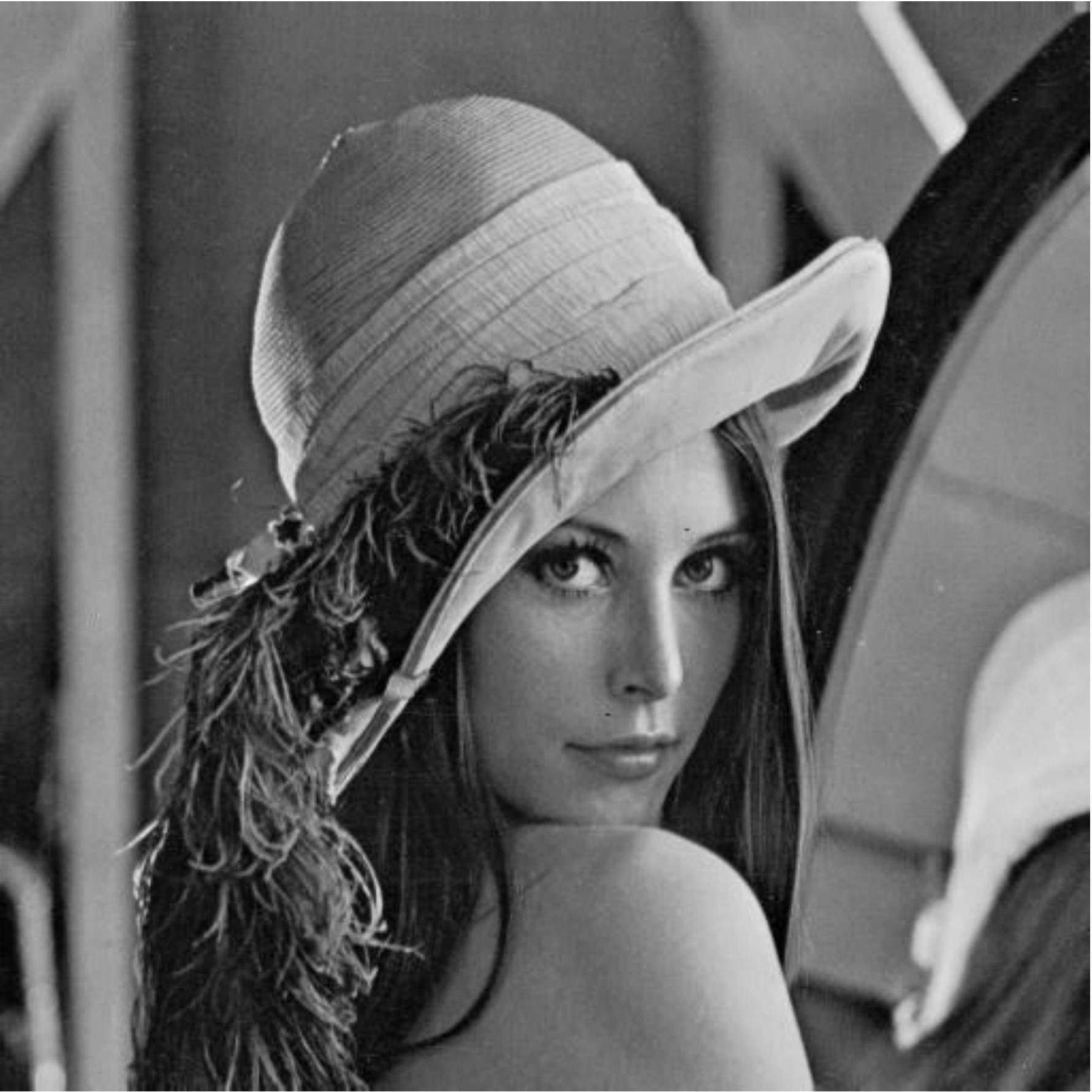}}
  \centerline{(a) Original Lena.}
\end{minipage}
\begin{minipage}[b]{.49\linewidth}
  \centering
    \centerline{\includegraphics[width=2.5cm]{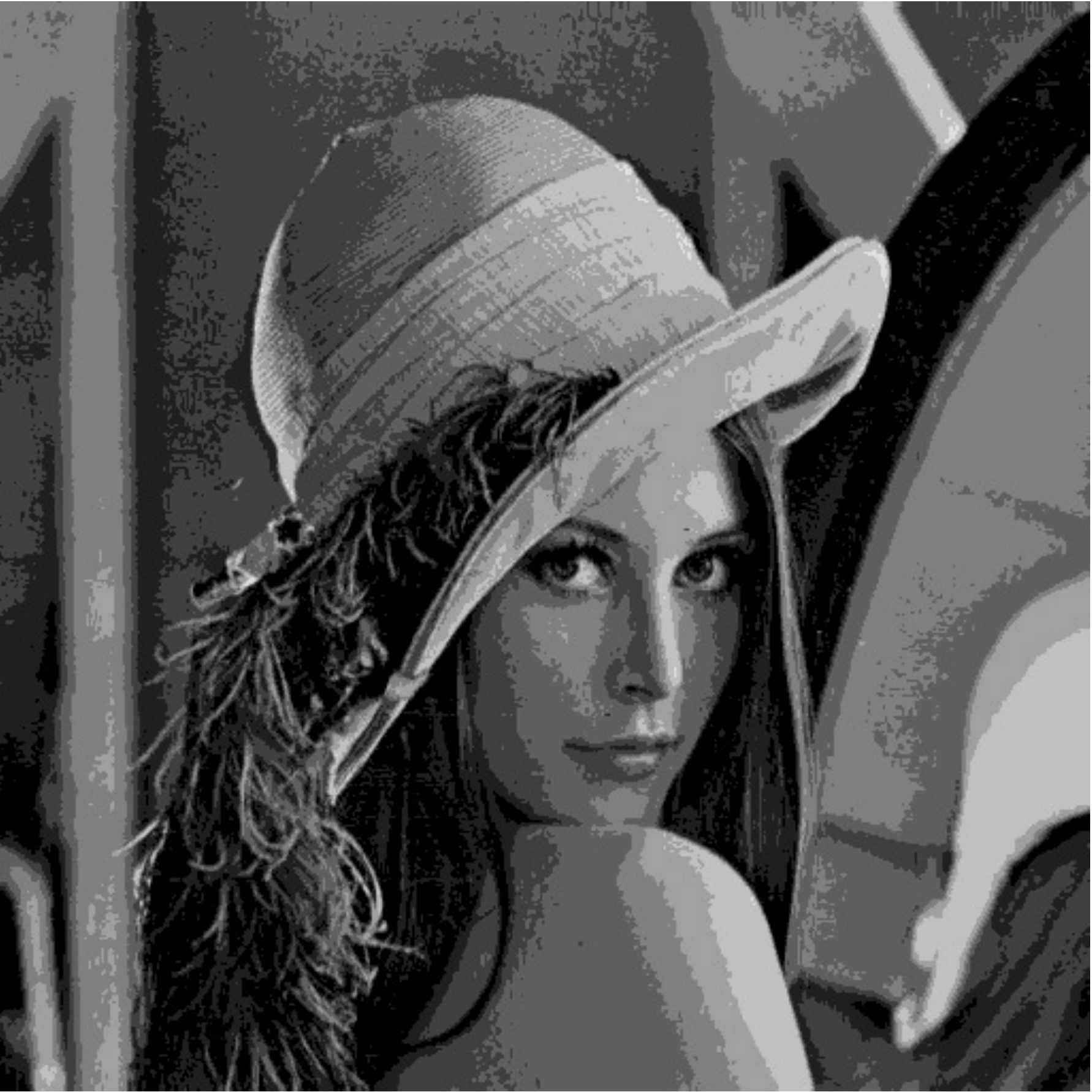}}
  \centerline{(b) MSCs of Lena.}
\end{minipage}
\hfill
\begin{minipage}[b]{0.49\linewidth}
  \centering
    \centerline{\includegraphics[width=2.5cm]{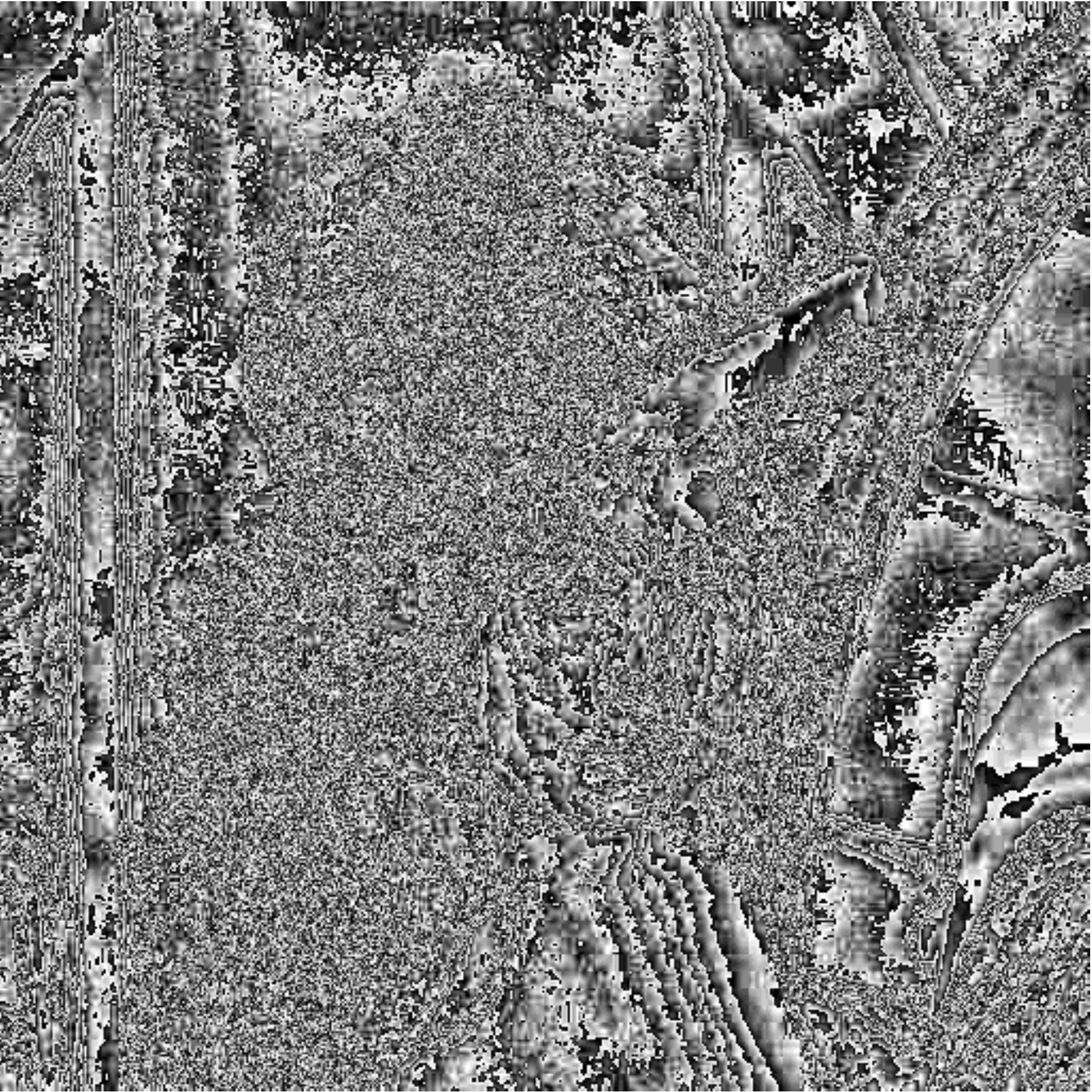}}
  
  \centerline{(c) LSCs of Lena ($\times 17$).}
\end{minipage}
\caption{Most and least significant coefficients of Lena.}
\label{fig:MSCLSC}
\end{figure}
\end{example}

\section{The new Process: $\mathcal{DI}_3$}\label{sec:new-process-di-1}


In this section, a new algorithm, which is inspired from the
scheme $\mathcal{CIS}_2$ described in~\cite{fgb11:ip}, is presented. 
It is easyer to implement for Internet applications, especially in order to
guarantee anonymization. Moreover, this new scheme $\mathcal{DI}_3$ seems to be
faster than $\mathcal{CIS}_2$, which is a major advantage to have fast response
times on the Internet.

Let us firstly introduce the following notations.
$\mathsf{P} \in \mathds{N}^\ast$ is the width, in term of bits, of the
  message to embed into the cover media.
 $\lambda \in \mathds{N}^{\ast}$ is the number of iterations to realize,
 which is s.t. $\lambda > \mathsf{P}$.
 $x^0 \in \mathbb{B}^\mathsf{N}$ is for the $\mathsf{N}$ LSCs of a given
  cover media $C$ supposed to be uniformly distributed.
 $m \in \mathbb{B}^\mathsf{P}$ is the message to hide into $x^0$.
Finally, $S \in \mathbb{S}_\mathsf{P}$ is a strategy such that the
  finite sequence $\left\{S^k, k \in \llbracket \lambda - \mathsf{P}
  +1;\lambda \rrbracket\right\}$ is injective.

\begin{remark}
The width $\mathsf{P}$ of the message to hide into the LSCs of the cover media
$x^0$ has to be far smaller than the number of LSCs.
\end{remark}


The proposed information hiding scheme is defined by:
\begin{definition}[$\mathcal{DI}_3$ Data hiding scheme]\ 
\label{def:process-DI-1}
$\forall (n,i,j) \in
\mathds{N}^{\ast} \times \llbracket 0;\mathsf{N-1}\rrbracket \times \llbracket
0;\mathsf{P-1}\rrbracket$:

\begin{equation*}
\begin{array}{l}
x_i^n=\left\{
\begin{array}{ll}
x_i^{n-1} & \text{ if }S^n\neq i \\
m_{S^n} & \text{ if }S^n=i.
\end{array}
\right.
\end{array}
\end{equation*}
\end{definition}

The stego-content is the Boolean vector $y = x^\lambda \in
\mathbb{B}^\mathsf{N}$, which will replace the former LSCs (LSCs of the cover media are replaced by the
vector $y$).
\section{Data Hiding Security and Robustness}
\label{sec:dh-security}

\subsection{Security and robustness}


Even if security and robustness are neighboring concepts without clearly
established definitions~\cite{Perez-Freire06}, robustness is often considered
to be mostly concerned with blind elementary attacks, whereas security is not
limited to certain specific attacks. Indeed, security encompasses robustness
and intentional attacks~\cite{Kalker2001,ComesanaPP05bis}. The best attempt to
give an elegant and concise definition for each of these two terms was
proposed in \cite{Kalker2001}. Following Kalker, we will consider in this
research work the two following definitions:  

\begin{definition}[Security~\cite{Kalker2001}]\label{def:security}
Watermarking security refers to the
inability by unauthorized users to have access to the raw watermarking channel
[...] to remove, detect and estimate, write or modify the raw watermarking
bits.
\end{definition}

\begin{definition}[Robustness~\cite{Kalker2001}]\label{def:robustness}
Robust watermarking is a mechanism to create a communication channel that is
multiplexed into original content [...] It is required that, firstly, the
perceptual degradation of the marked content [...] is minimal and, secondly,
that the capacity of the watermark channel degrades as a smooth function of the
degradation of the marked content. 
\end{definition}

In this article, we will focus more specifically on the security aspects, which have been formalized in the 
Simmons' prisoner problem.

\subsection{The prisoner problem}

In the prisoner problem of Simmons~\cite{Simmons83}, Alice and Bob are in jail,
and they want to, possibly, devise an escape plan by exchanging hidden messages
in innocent-looking cover contents (Fig.~\ref{fig:simmons-prisonner-problem}).
These messages are to be conveyed to one another by a common warden, Eve, who
over-drops all contents and can choose to interrupt the communication if they
appear to be stego-contents.

\begin{figure}
\begin{center}
\includegraphics[width=8cm]{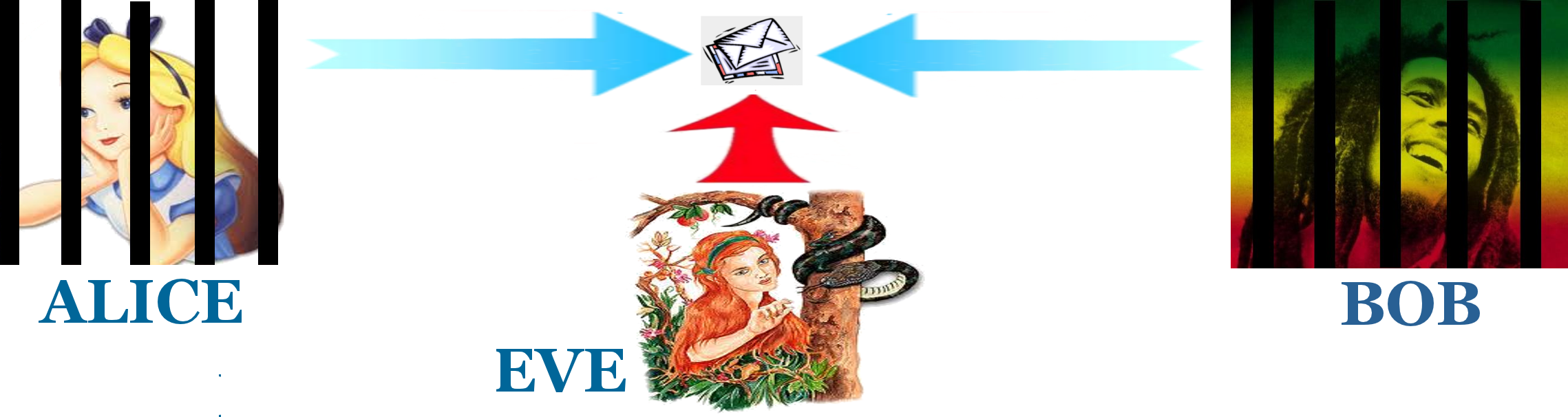}
\end{center}
\caption{Simmons' prisoner problem~\cite{Simmons83}}
\label{fig:simmons-prisonner-problem}
\end{figure}

\subsection{Classification of Attacks}\label{sec:attack-classes}

In the steganography framework, in the Simmons' prisoner problem context,   attacks have been classified in~\cite{Cayre2008} as follows.

%
%
%

\begin{definition}[Classes of attacks]\ 
\begin{description}
\item[WOA:] A \emph{Watermark-Only Attack} occurs when an attacker has only
access to several watermarked contents.

\item[KMA:] A \emph{Known-Message Attack} occurs when an attacker has access to
several pairs of watermarked contents and corresponding hidden messages.

\item[KOA:] A \emph{Known-Original Attack} is when an attacker has access to
several pairs of watermarked contents and their corresponding original versions.

\item[CMA:] A \emph{Constant-Message Attack} occurs when the attacker observes
several watermarked contents and only knows that the unknown hidden message is the same in all contents.
\end{description}
\end{definition}

A synthesis of this classification is given in
Table~\ref{table:attack-classification}.
\begin{table}[h]
\begin{center}

 \begin{tabular}{|c||c|c|c|}
  \hline
   \textbf{Class} &  \textbf{Original content} &  \textbf{Stego content} & 
   \textbf{Hidden message}\\
    \hline 
    \hline 
 \textbf{WOA} &   &    $\times$ & \\
  \hline
 \textbf{KMA}  & & $\times$ & $\times$ \\
  \hline
 \textbf{KOA} & $\times$ & $\times$ & \\
 \hline
 \textbf{CMA} & & & $\times$ \\
  \hline 
  \end{tabular}
  \caption{Watermarking attacks classification in context of~\cite{Kalker2001}}
  \label{table:attack-classification}
  \end{center}
  \end{table}
 
In this article, we will focus more specifically on the ``Watermark-Only
Attack'' situation, which is the most relevant category when considering
 anonymity and privacy protection through the Internet.

\subsection{Reminder about Stego-Security}\label{sec:stego-security}

 The stego-security, defined in the 
\emph{Watermark-Only Attack} (WOA) framework, is the highest security level that can be defined in this
setup~\cite{Cayre2008}.

\begin{definition}[Stego-Security]
\label{Def:Stego-security}\  
Let $\mathds{K}$ be the set of embedding keys, $p(X)$ the probabilistic model of $N_0$ initial host contents, and $p(Y|K_1)$  the probabilistic model of $N_0$ watermarked contents.
Moreover, each host content has been
watermarked with the same secret key $K_1$ and the same embedding function $e$.
Then $e$ is said \emph{stego-secure} if:
$$\forall K_1 \in \mathds{K}, p(Y|K_1)=p(X).$$
\end{definition}

Until now, only three schemes have been proven 
stego-secure. On the one hand, the authors of \cite{Cayre2008} have established
that the spread spectrum technique called Natural Watermarking is stego-secure
when its distortion parameter $\eta$ is equal to $1$. On the other hand, we have proposed in~\cite{gfb10:ip} and~\cite{fgb11:ip} two other data hiding
schemes satisfying this security property.



\section{Security Study}

Let us prove that,

\begin{proposition}
$\mathcal{DI}_3$ is stego-secure.
\end{proposition}

\begin{proof}
Let us suppose that $x^0 \sim
\mathbf{U}\left(\mathbb{B}^\mathsf{N}\right)$, $m \sim
\mathbf{U}\left(\mathbb{B}^\mathsf{P}\right)$, and $S \sim
\mathbf{U}\left(\mathbb{S}_\mathsf{P}\right)$ in a $\mathcal{DI}_3$ setup, where $\mathbf{U}(X)$ describes
the uniform distribution on $X$. We
will prove by a mathematical induction that $\forall n \in \mathds{N}, x^n \sim \mathbf{U}\left(\mathbb{B}^\mathsf{N}\right)$. The base case is obvious
according to the uniform repartition hypothesis. 

Let us now suppose that the statement $x^n \sim
\mathbf{U}\left(\mathbb{B}^\mathsf{N}\right)$ holds for some $n$ (
$P\left(x^{n}=k\right)= \frac{1}{2^N}$).

 For a given $k  \in \mathbb{B}^N$,
 we denote by $\tilde{k_i} \in \mathbb{B}^N$ the vector defined by:\newline
 $\forall i \in  \llbracket 0;\mathsf{N-1}\rrbracket,$
 if
 $k=\left(k_0,k_1,\hdots,k_i,\hdots,k_{\mathsf{N}-2},k_{\mathsf{N}-1}\right)$,\newline
 then $\tilde{k}_i=\left( k_0,k_1,\hdots,\overline{k_i},\hdots,k_{\mathsf{N}-2},k_{\mathsf{N}-1} \right)$, where
 $\overline{x}$ is the negation of the bit $x$.

Let $p$ be defined by: $p=P\left(x^{n+1}=k\right)$.
Let
$E_{j}$ and $E$ be the events defined by:
$\forall j \in \llbracket
0;\mathsf{P-1}\rrbracket, E_j=
(x^n = \tilde{k_j}) \wedge (S^n = j) \wedge (m_{S^n} = k_j),$$
$$E = (x^n = k) \wedge (m_{S^n} =
x_{S^n}).$ So, $p=P\left(E \vee \bigvee_{j = 0}^{\mathsf{N-1}}
E_j \right).$

On the one hand, $\forall j \in \llbracket 0;\mathsf{P-1}\rrbracket,$ the event
$E_j$ is a conjunction of the sub-events $(S^n = j)$ and other sub-events. $\forall j \in
\llbracket 0;\mathsf{P-1}\rrbracket,$ all the sub-events $(S^n = j)$ are clearly
pairwise disjoints, so all the evente $E_j$ are pairwise disjoints too.

On the other hand, $\forall j \in \llbracket 0;\mathsf{P-1}\rrbracket,$ the
events $E_j$ and $E$ are disjoints, because in $E_j$,  a conjunction
of the sub-event $(x^n = \tilde{k_j})$ with other sub-events appears, whereas in $E$ 
a conjunction of the sub-event $(x^n = k)$ with other sub-events appears, and the two  sub-events
$(x^n = \tilde{k_j})$ and $(x^n = k)$ are clearly disjoints.

As a consequence, using the probability law concerning the reunion of disjoint
events we can claim that: $p=P(E) + \sum_{j=0}^N P(E_j)$.

Now we evaluate both $P(E)$ and $P(E_j)$.

\begin{enumerate}
  \item \emph{The case of $P(E)$}: As the two events $(x^n = k)$ and $(m_{S^n}
  = x_{S^n})$ concern two different sequences, they are clearly independent.

Then, by using the inductive hypothesis: $P(x^n = k) = \frac{1}{2^N}$. So,

\begin{equation*}
\begin{array}{ccl}
p(E) & = & P(x^n = k)  \times  P(m_{S^n} = x_{S^n})\\
& = & \frac{1}{2^\mathsf{N}}  \times  \left[P(m_{S^n} = 0)P(x_{S^n} =
0)\right.\\
& &+  \left.P(m_{S^n} = 1) P(x_{S^n} = 1)\right]\\
& = & \frac{1}{2^\mathsf{N}}  \times  \left[P(m_{S^n} = 0)P(x_{S^n} =
0)\right.\\
& & \left. +  P(m_{S^n} = 1) (1 - P(x_{S^n} = 0))\right]\\
& = & \frac{1}{2^\mathsf{N}}  \times  \left[\frac{1}{2}P(x_{S^n} = 0) +
\frac{1}{2} (1 - P(x_{S^n} = 0))\right]\\
& = & \frac{1}{2^{\mathsf{N}+1}}.   \\
\end{array}
\end{equation*}

  \item \emph{Evaluation of $P(E_j)$}: As the three events $(x^n =
  \tilde{k_j})$, $(S^n = j)$, and $(m_n = k_j)$ deal with three different sequences, they are clearly independent. So
\begin{equation*}
\begin{array}{ccl}
P(E_j) & = & P(x^n = \tilde{k_j})  \times  P(S^n = j) \times  P(m_{S^n} = k_j)\\
 & = & \frac{1}{2^\mathsf{N}}  \times  \frac{1}{\mathsf{P}}  \times
 \frac{1}{2}\\
  &=& \frac{1}{\mathsf{P}}  \times  \frac{1}{2^\mathsf{N+1}}, \\
\end{array}
\end{equation*}
\end{enumerate}
due to the hypothesis of uniform repartition of $S$ and $m$.
\begin{equation*}
\begin{array}{lccccc}
\text{Consequently,}&p&=&P(E) &+ &\sum_{j=0}^\mathsf{P} P(E_j)\\
&&=&\frac{1}{2^\mathsf{N+1}} & + &\sum_{j=0}^\mathsf{P}
\left(\frac{1}{\mathsf{P}} \times \frac{1}{2^\mathsf{N+1}}\right)\\
&&=&\frac{1}{2^\mathsf{N}}.&&\\
\end{array}
\end{equation*}

Finally, 
$P\left(x^{n+1}=k\right)=\frac{1}{2^N}$, which leads to $x^{n+1} \sim
\mathbf{U}\left(\mathbb{B}^\mathsf{N}\right)$.
This result is true $\forall n \in \mathds{N}$, we thus have proven that the
stego-content $y$ is uniformly distributed in the set of possible stego-contents:
$y \sim \mathbf{U}\left(\mathbb{B}^\mathsf{N}\right) \text{
when } x \sim \mathbf{U}\left(\mathbb{B}^\mathsf{N}\right).$
\end{proof}

\begin{remark}[Distribution of LSCs]
We have supposed that $x^0 \sim
\mathbf{U}\left(\mathbb{B}^\mathsf{N}\right)$
to prove the stego-security of the data hiding process
$\mathcal{DI}_3$.
This hypothesis is the most restrictive one, but it can 
be obtained at least partially in two possible manners. 
Either a channel that appears to be random (for instance, when
applying a chi squared test) can be found in the media. Or a systematic process
can be applied on the images to obtain this uniformity, as follows.
Before embedding the hidden message, all the original LSCs must be replaced by randomly
generated ones, hoping so that such cover media will be considered to be noisy by any given attacker.

Let us remark that, in
the field of data anonymity for privacy on the 
Internet,  we are in the
``watermark-only attack'' framework.
As it has been recalled  in
Table~\ref{table:attack-classification},
in that framework,
the attacker has only access to stego-contents, having so no knowledge
of the original media, before introducing the message in the random channel (LSCs).
However, this assumption of the existence of a random channel,
natural or artificial, into the cover images, is clearly
the most disputable one of this research work.
The authors' intention is to investigate such hypothesis
more largely in future works, by investigation the
distribution of several LSCs contained in a
large variety of images chosen randomly on the Internet. Among other things, we will check if some well-defined
LSCs are naturally uniformly distributed in most cases. 
To conduct such studies, we intend to use the well-known NIST (National Institute of Standards and Technology of the U.S.
Government) tests suite~\cite{ANDREW2008}, the
DieHARD~\cite{Marsaglia1996} battery, or the stringent TestU01~\cite{L'ecuyer2009}. 
Depending on the
results of this search for randomness in natural images,
 the need to introduce an artificial random channel could be possibly removed.
\end{remark}

\begin{remark}[Distribution of the messages $m$]
 In order to prove the stego-security of the data hiding
 process
$\mathcal{DI}_3$, we have supposed that  $m \sim
\mathbf{U}\left(\mathbb{B}^\mathsf{P}\right)$. This hypothesis is not
really restrictive. Indeed, to encrypt the message before its
embedding into the LSCs of cover media if sufficient to achieve this goal. To say it different, in order to be in the conditions of applications of the process
$\mathcal{DI}_3$, the hidden message must be encrypted.
\end{remark}

\begin{remark}[Distribution of the strategies $S$]
 To prove the stego-security of the data hiding
 process
$\mathcal{DI}_3$, we have finally supposed that  $S \sim
\mathbf{U}\left(\mathbb{S}_\mathsf{P}\right)$.  This hypothesis is not
restrictive too, as any cryptographically secure
pseudorandom generator (PRNG)
satisfies this property. With such PRNGs, it is impossible in 
polynomial time, to make the distinction between  random numbers and  numbers provided by
these generators. For instance, \emph{Blum Blum Shub
(BBS)}~\cite{junod1999cryptographic}, \emph{Blum Goldwasser (BG)}~\cite{prng-Blum-goldwasser}, 
or \emph{ISAAC}~\cite{prng-isaac},  are convenient here.
\end{remark}

\section{Steganalysis}\label{section:steganalysis}

The steganographic scheme detailed along these lines has been compared to 
state of the art steganographic approaches, namely
YASS~\cite{DBLP:conf/ih/SolankiSM07},
HUGO~\cite{DBLP:conf/ih/PevnyFB10}, and
nsF5~\cite{DBLP:conf/mmsec/FridrichPK07}.

The steganalysis is based on the BOSS image
database~\cite{DBLP:conf/ih/BasFP11} 
which consists in a set of 10 000 512x512 greyscale images. 
We randomly selected 50 of them to compute the cover set.
Since YASS and nsF5 are dedicated to JPEG 
support, all these images have been firstly translated into JPEG format
thanks to the \verb+mogrify+ command line. 
To allow the comparison between steganographic schemes, the relative payload
is always set with 0.1 bit per pixel. Under that constrain,
the embedded message $m$ is a sequence of 26214 randomly generated bits. 
This step has led to distinguish four sets of stego contents, one for each steganographic approach.

Next we use the steganalysis tool developed by the HugoBreakers 
team~\cite{holmes11,ensemble11} based on AI classifier and which won 
the BOSS competition~\cite{DBLP:conf/ih/BasFP11}. 
Table~\ref{table:holme} summarizes these steganalysis results expressed 
as the error probabilities of the steganalyser. The errors are the mean 
of the false alarms and of the missed detections. An error that is closed 
to 0.5 signifies that deciding whether an image contains 
a stego content is a random choice for the steganalyser. 
Conversely, a tiny error denotes that the steganalyser can easily classify
stego content and non stego content.   

\begin{table}[ht]
\begin{center}
\begin{tabular}{|l|l|l|l|l|}
\hline
Steganographic Tool & $\mathcal{DI}_1$ & YASS & HUGO & NsF5 \\
\hline
Error Probability   & 0.4133& 0.0067& 0.495 & 0.47 \\
\hline
\end{tabular}
\end{center}
\caption{Steganalysis results of HugoBreakers steganalyser applied on 
steganographic scheme}\label{table:holme}
\end{table}


The best result is obtained by HUGO, which is closed to the perfect
steganographic approach to the considered steganalyser, since the error is about 0.5. However, even if 
the approach detailed along these lines has not any optimization, these 
first experiments show promising results.
We finally notice that the HugoBreakers's steganalyser should outperform
these results  on larger image databases, \textit{e.g.}, when applied on the
whole BOSS image database.

%

\section{Conclusion and Future Work}

Steganography is a real alternative to
guarantee anonymity through the Internet. For instance, 
the scheme presented in this article
offers a secure solution to achieve this goal, thanks to its stego-security. Even if this new scheme
$\mathcal{DI}_3$ does not possess topological properties (unlike the
$\mathcal{CIS}_2$), its level of security seems to be sufficient for
Internet applications. 
Indeed, we take place into the \emph{Watermark Only Attack (WOA)} framework, where stego-security 
is the highest level of security.
Additionally, this new scheme is faster than
$\mathcal{CIS}_2$. This is a major advantage for an utilization through
the Internet, to respect response times of web sites. 
Moreover, for this first version of the process, the  steganalysis results
are promising. 

In future work, various improvements of this scheme are planed to obtain better scores against steganalysers. 
For instance, LSCs will be embedded into various frequency domains. 
The robustness of the proposed scheme will be evaluated
too~\cite{bcg11:ij}, to determine whether this
information hiding algorithm can be relevant in other
Internet domains interesting by data hiding techniques,
as the semantic web.
Finally a cryptographic
approach of information hiding security is currently investigated, enlarging the Simmons' prisoner problem~\cite{arxiv-bgh2012-DBLP:journals/corr/abs-1112-5245}, and we intend  to evaluate the proposed
scheme in this framework.


\bibliographystyle{plain}
\bibliography{jabref}

\begin{thebibliography}{10}

\bibitem{bcg11:ij}
Jacques Bahi, Jean-Fran\c{c}ois Couchot, and Christophe Guyeux.
\newblock Steganography: a class of secure and robust algorithms.
\newblock {\em The Computer Journal}, pages ***--***, 2011.
\newblock Available online. Paper version to appear.

\bibitem{bcgw11:ip}
Jacques Bahi, Jean-Fran\c{c}ois Couchot, Christophe Guyeux, and Qianxue Wang.
\newblock Class of trustworthy pseudo random number generators.
\newblock In {\em INTERNET 2011, the 3-rd Int. Conf. on Evolving Internet},
  pages ***--***, Luxembourg, Luxembourg, June 2011.
\newblock To appear.

\bibitem{arxiv-bgh2012-DBLP:journals/corr/abs-1112-5245}
Jacques~M. Bahi, Christophe Guyeux, and Pierre-Cyrille H{\'e}am.
\newblock A complexity approach for steganalysis.
\newblock {\em CoRR}, abs/1112.5245, 2011.

\bibitem{DBLP:conf/ih/BasFP11}
P.~Bas, T.~Filler, and T.~Pevn\'{y}.
\newblock Break our steganographic system --- the ins and outs of organizing
  boss.
\newblock In T.~Filler, editor, {\em Information Hiding, 13th International
  Workshop}, Lecture Notes in Computer Science, Prague, Czech Republic, May
  18--20, 2011. Springer-Verlag, New York.

\bibitem{Cayre2008}
Francois Cayre, Caroline Fontaine, and Teddy Furon.
\newblock Kerckhoffs-based embedding security classes for woa data hiding.
\newblock {\em IEEE Transactions on Information Forensics and Security},
  3(1):1--15, 2008.

\bibitem{Clark:2007:UAW:1280680.1280687}
Jeremy Clark, P.~C. van Oorschot, and Carlisle Adams.
\newblock Usability of anonymous web browsing: an examination of tor interfaces
  and deployability.
\newblock In {\em Proceedings of the 3rd symposium on Usable privacy and
  security}, SOUPS '07, pages 41--51, New York, NY, USA, 2007. ACM.

\bibitem{ComesanaPP05bis}
Pedro Comesa{\~n}a, Luis P{\'e}rez-Freire, and Fernando P{\'e}rez-Gonz{\'a}lez.
\newblock Fundamentals of data hiding security and their application to
  spread-spectrum analysis.
\newblock In Mauro Barni, Jordi Herrera-Joancomart\'{\i}, Stefan Katzenbeisser,
  and Fernando P{\'e}rez-Gonz{\'a}lez, editors, {\em IH'05: Information Hiding
  Workshop}, volume 3727 of {\em Lecture Notes in Computer Science}, pages
  146--160. Lectures Notes in Computer Science, Springer-Verlag, 2005.

\bibitem{www:perseus}
ESIEA.
\newblock Technologie perseus pour l'anonymisation sur internet, 02 2012.
\newblock [On line - 2012.02.22].

\bibitem{www:perseus-firefox-plugin}
ESIEA.
\newblock Technologie perseus pour l'anonymisation sur internet - plugin
  firefox, 02 2012.
\newblock [On line - 2012.02.22].

\bibitem{DBLP:conf/mmsec/FridrichPK07}
Jessica~J. Fridrich, Tom{\'a}s Pevn{\'y}, and Jan Kodovsk{\'y}.
\newblock Statistically undetectable jpeg steganography: dead ends challenges,
  and opportunities.
\newblock In Deepa Kundur, Balakrishnan Prabhakaran, Jana Dittmann, and
  Jessica~J. Fridrich, editors, {\em MM{\&}Sec}, pages 3--14. ACM, 2007.

\bibitem{fgb11:ip}
Nicolas Friot, Christophe Guyeux, and Jacques Bahi.
\newblock Chaotic iterations for steganography - stego-security and
  chaos-security.
\newblock In {\em SECRYPT'2011, Int. Conf. on Security and Cryptography}, pages
  ***--***, Sevilla, Spain, July 2011.
\newblock To appear.

\bibitem{bg10:ip}
Christophe Guyeux and Jacques Bahi.
\newblock An improved watermarking algorithm for internet applications.
\newblock In {\em INTERNET'2010. The 2nd Int. Conf. on Evolving Internet},
  pages 119--124, Valencia, Spain, September 2010.

\bibitem{gfb10:ip}
Christophe Guyeux, Nicolas Friot, and Jacques Bahi.
\newblock Chaotic iterations versus spread-spectrum: chaos and stego security.
\newblock In {\em IIH-MSP'10, 6-th Int. Conf. on Intelligent Information Hiding
  and Multimedia Signal Processing}, pages 208--211, Darmstadt, Germany,
  October 2010.

\bibitem{prng-isaac}
Robert Jenkins.
\newblock Isaac.
\newblock In Dieter Gollmann, editor, {\em Fast Software Encryption}, volume
  1039 of {\em Lecture Notes in Computer Science}, pages 41--49. Springer
  Berlin / Heidelberg, 1996.
\newblock 10.1007/3-540-60865-6\_41.

\bibitem{junod1999cryptographic}
P.~Junod.
\newblock {\em Cryptographic secure pseudo-random bits generation: The
  Blum-Blum-Shub generator}.
\newblock August, 1999.

\bibitem{Kalker2001}
T.~Kalker.
\newblock Considerations on watermarking security.
\newblock In {\em Multimedia Signal Processing, 2001 IEEE Fourth Workshop on},
  pages 201--206, 2001.

\bibitem{holmes11}
J.~Kodovský and J.~Fridrich.
\newblock Steganalysis in high dimensions: fusing classifiers built on random
  subspaces.
\newblock In {\em Proc. SPIE, Electronic Imaging, Media Watermarking, Security,
  and Forensics XIII, San Francisco, CA,}, January 2011.

\bibitem{ensemble11}
J.~Kodovský, J.~Fridrich, and V.~Holub.
\newblock Ensemble classifiers for steganalysis of digital media.
\newblock {\em IEEE Transactions on Information Forensics and Security}, PP
  Issue:99:1 -- 1, 2011.
\newblock To appear.

\bibitem{L'ecuyer2009}
P.~L'ecuyer and R.~Simard.
\newblock Testu01: A software library in ansi c for empirical testing of random
  number generators.
\newblock {\em Laboratoire de simulation et d'optimisation. Universit\'{e} de
  Montr\'{e}al IRO}, 2009.

\bibitem{Marsaglia1996}
G.~Marsaglia.
\newblock Diehard: a battery of tests of randomness.
\newblock {\em http://stat.fsu.edu/~geo/diehard.html}, 1996.

\bibitem{Perez-Freire06}
Luis Perez-Freire, Pedro Comesana, Juan~Ramon Troncoso-Pastoriza, and Fernando
  Perez-Gonzalez.
\newblock Watermarking security: a survey.
\newblock In {\em LNCS Transactions on Data Hiding and Multimedia Security},
  2006.

\bibitem{DBLP:conf/ih/PevnyFB10}
Tom{\'a}s Pevn{\'y}, Tom{\'a}s Filler, and Patrick Bas.
\newblock Using high-dimensional image models to perform highly undetectable
  steganography.
\newblock In Rainer B{\"o}hme, Philip W.~L. Fong, and Reihaneh Safavi-Naini,
  editors, {\em Information Hiding}, volume 6387 of {\em Lecture Notes in
  Computer Science}, pages 161--177. Springer, 2010.

\bibitem{ANDREW2008}
NIST Special Publication 800-22 rev. 1.
\newblock A statistical test suite for random and pseudorandom number
  generators for cryptographic applications.
\newblock NIST, August 2008.

\bibitem{Simmons83}
Gustavus~J. Simmons.
\newblock The prisoners' problem and the subliminal channel.
\newblock In {\em Advances in Cryptology, Proc. CRYPTO'83}, pages 51--67, 1984.

\bibitem{DBLP:conf/ih/SolankiSM07}
Kaushal Solanki, Anindya Sarkar, and B.~S. Manjunath.
\newblock Yass: Yet another steganographic scheme that resists blind
  steganalysis.
\newblock In Teddy Furon, François Cayre, Gwenaël~J. Doërr, and Patrick Bas,
  editors, {\em Information Hiding}, volume 4567 of {\em Lecture Notes in
  Computer Science}, pages 16--31. Springer, 2007.

\bibitem{prng-Blum-goldwasser}
Umesh Vazirani and Vijay Vazirani.
\newblock Efficient and secure pseudo-random number generation (extended
  abstract).
\newblock In George Blakley and David Chaum, editors, {\em Advances in
  Cryptology}, volume 196 of {\em Lecture Notes in Computer Science}, pages
  193--202. Springer Berlin / Heidelberg, 1985.
\newblock 10.1007/3-540-39568-7\_17.

\bibitem{wiki:Hossein-Derakhshan}
Wikipédia.
\newblock Hossein derakhshan --- wikipédia{,} l'encyclopédie libre, 2011.
\newblock [En ligne; Page disponible le 22-février-2012].

\bibitem{wiki:Vahid-Asghari}
Wikipédia.
\newblock Vahid asghari --- wikipédia{,} l'encyclopédie libre, 2011.
\newblock [En ligne; Page disponible le 22-février-2012].

\bibitem{www:tor}
www.
\newblock Tor: Anonymity online - protect your privacy. defend yourself against
  network surveillance and traffic analysis., 02 2012.
\newblock [On line - 2012.02.22].

\end{thebibliography}

\end{document}